\documentclass[letterpaper,11pt]{article}
\usepackage{amsmath,amsxtra,amssymb,amsthm,amsfonts
}
\usepackage{mathrsfs}
\usepackage[T1]{fontenc}
\usepackage[utf8]{inputenc}
\usepackage{mathtools}   % loads »amsmath«

\usepackage{color,hyperref}
\numberwithin{equation}{section}
\newtheorem{lemma}{Lemma}[section]

\newtheorem{theorem}[lemma]{Theorem}

\newtheorem{prob}[lemma]{Problem}
\newtheorem{rem}[lemma]{Remark}
\newtheorem{remark}[lemma]{Remark}

\newtheorem{definition}[lemma]{Definition}

\newcommand{\re}{\begin{rem}\rm}
  \newcommand{\mar}{\end{rem}}

\newcommand{\ee }{\mathrm{I}\!\!1}

\newcommand{\mel}{\end{eqnarray*}}

\newcommand{\pl}{\hspace{.1cm}}

\newcommand{\qd}{\end{proof}\vspace{0.5ex}}

\newcommand{\pf}{\begin{proof}}

\newcommand{\be}{\left|{\atop}}

\newcommand{\xspace}{\hbox{\kern-2.5pt}}
\newcommand{\xyspace}{\hbox{\kern-1.1pt}}

\newcommand\bra[1]{\langle  #1|}
\newcommand\ket[1]{| #1\rangle}

\allowdisplaybreaks

\definecolor{LightGray}{rgb}{0.94,0.94,0.94}
\definecolor{VeryLightBlue}{rgb}{0.9,0.9,1}
\definecolor{LightBlue}{rgb}{0.8,0.8,1}
\definecolor{DarkBlue}{rgb}{0,0,0.6}
\definecolor{LightGreen}{rgb}{0.88,1,0.88}
\definecolor{MidGreen}{rgb}{0.6,1,0.6}
\definecolor{DarkGreen}{rgb}{0,0.6,0}
\definecolor{DarkGrreen}{rgb}{0,0.8,0}

\definecolor{VeryLightYellow}{rgb}{1,1,0.9}
\definecolor{LightYellow}{rgb}{1,1,0.6}
\definecolor{MidYellow}{rgb}{1,1,0.5}
\definecolor{DarkYellow}{rgb}{0.8,1,0.3}
\definecolor{VeryLightRed}{rgb}{1,0.9,0.9}
\definecolor{LightRed}{rgb}{1,0.8,0.8}
\definecolor{DarkRed}{rgb}{0.8,0.2,0}
\definecolor{DarkRedb}{rgb}{0.6,0.2,0}
\definecolor{DarkLila}{rgb}{0.8,0,1}
\definecolor{Beige}{rgb}{0.96,0.96,0.86}
\definecolor{Gold}{rgb}{1.,0.84,0.}
\definecolor{Goldb}{rgb}{0.7,0.3,0.5}
\definecolor{MyYellow}{rgb}{1.,0.84,0.8}

\def\11{\mathbb{I}}

\usepackage{graphicx}

\DeclareRobustCommand\openone{\leavevmode\hbox{\small1\normalsize\kern-.33em1}}

\renewcommand{\be}{\begin{equation}}
	\renewcommand{\ee}{\end{equation}}
\newcommand{\bea}{\begin{eqnarray}}
	\newcommand{\eea}{\end{eqnarray}}
\newcommand{\beas}{\begin{eqnarray*}}
	\newcommand{\eeas}{\end{eqnarray*}}

\usepackage{amsmath}
\usepackage{ams math}
\usepackage{amssymb}
\usepackage{amsthm}
\usepackage{dsfont}
\usepackage{upgreek}
\usepackage{enumitem}
\usepackage{array}
\usepackage{ams fonts}
\usepackage{graphics}
\usepackage[utf8]{inputenc}
\newtheorem*{theorem*}{Theorem}
\newtheorem*{remark*}{Remark}
\newtheorem*{lemma*}{Lemma}
\newtheorem*{note*}{Note}
\newtheorem*{prop*}{Proposition}
\newtheorem*{fact*}{Fact}

\newcommand{\tr}{\mbox{tr}}

\usepackage[margin=1.0 in]{geometry}
\usepackage{tikz-cd}
\providecommand{\keywords}[1]
{
  \textbf{\textit{Keywords: }} #1
}
\usepackage{authblk}
\begin{document}

\title{Monotonicity of optimized quantum $f$-divergence}
\author[1]{Haojian Li\thanks{Haojian Li: lihaojianmath@gmail.com}}
\affil[1]{Zentrum Mathematik\\ Technische Universit\"at M\"unchen, Garching, 85748, Germany}
% \email[Haojian Li]{lihaojianmath@gmail.com}
\maketitle
\begin{abstract} Optimized quantum $f$-divergence was first introduced by Wilde in \cite{Wil18}. Wilde raised the question of whether the monotonicity of optimized quantum $f$-divergence can be generalized to maps that are not quantum channels. We answer this question by generalizing the monotonicity of optimized quantum $f$-divergences to  positive trace preserving maps satisfying a Schwarz inequality.\\
\end{abstract}

\section{introduction}
%\textcolor{red}{say something about relative entropy}
Umegaki divergence is a fundamental concept in quantum information theory and quantum computation. It measures the distinguishability of two quantum states first introduced in \cite{um62}, also see the survey \cite{vedral2002role} and references therein. Various generalizations of Umegaki divergences have been introduced and extensively studied in the past few decades.  Petz defined the quasi quantum divergence via relative modular operator, see \cite{Petz86}. Wilde (\cite{Wil18}) 
introduced the optimized quantum $f$-divergence that further generalized the definition by Petz.
Another notable generalization is sandwiched $\alpha$-R\'enyi divergence, introduced independently by Wilde et al  (\cite{wilde2014strong}) and M\"uller-Lennert et al (\cite{muller2013quantum}). Very recently Hirche and Tomamichel (\cite{hirche2023quantum}) introduced a new family of quantum $f$-divergence by using the quantum version of hockey-stick divergence.

A crucial property of Umegaki divergence is monotonicity under the actions of quantum operations (also referred to as \textit{data processing inequality} in the literature), stating that the distinguishability of two quantum states does not increase after undergoing a quantum channel. Monotonicity plays an important role in the study of quantum channel capacity, quantum machine learning, quantum hypothesis testing, and etc. The aforementioned generalizations of Umegaki divergence all satisfy the monotonicity under the actions of quantum channels. It is noteworthy that the monotonicity of Umegaki divergence and sandwiched $\alpha$-R\'enyi divergence has been generalized to the positive trace preserving maps in \cite{muller2017monotonicity} based on \cite{beigi2013sandwiched}.

Wilde (\cite{Wil18})introduced the optimized quantum $f$-divergence as a unified definition of various divergences, including Umegaki divergence and sandwiched $\alpha$-R\'enyi divergence. The data processing inequality of  optimized quantum $f$-divergence was proved in \cite{Wil18} by demonstrating invariance under isometries and monotonicity under taking a partial trace. The recoverability and extension of optimized quantum $f$-divergence to the general von Neumann algebraic setting was accomplished in  \cite{GW20}.  In this paper, we prove the monotonicity of optimized quantum $f$-divergence for positive trace preserving maps satisfying a Schwarz inequality. It remains open whether the optimized quantum $f$-divergence is monotone under the actions of any positive trace preserving maps. Hirche and Tomamichel (\cite{hirche2023quantum}) defined the quantum $f$-divergence by considering the integral of the hockey-stick divergence. This definition is closely related to the sandwiched $\alpha$-R\'enyi divergence and Petz $\alpha$-R\'enyi divergence. 

%\section{Preliminary}definition and tpp monotonicity.\begin{definition}operator monotonicity and operator concavity/convexity\end{definition}\begin{definition}\end{definition} \begin{theorem}[Loewner-Heinz] \begin{enumerate} \item For $p\in[-1,0]$, the function $f(x)=x^{p}$ is operator monotone decreasing and operator concave.\item For $p\in[0,1]$, the function $f(x)=x^{p}$ is operator monotone and operator concave.\item For $p\in[1,2]$, the function $f(x)=x^{p}$ is operator convex.\end{enumerate}\end{theorem}

\noindent {\bf Notations.} We use $B(H)$ for the linear space of bounded linear operators defined in the complex Hilbert space $H$. We use $B_{sa}(H)\subset B(H)$ for the space of self-adjoint operators. We use $B_{+}(H)\subset B(H)$ for the space of positive definite operators. We use $\tr$ as the trace on $B(H)$ and the $\langle A,B\rangle=\tr(A^{*}B)$ as the Hilbert-Schmidt inner product. 

\section{Monotonicity of quantum optimized $f$-divergence}
Let us recall the definition of the optimized quantum $f$-divergence, see \cite{Wil18} for more explanation and properties.
\begin{definition} Let $f$ be a function with domain $(0,\infty)$ and range $\mathbb{R}$. For positive semi-definite operators $\rho,\sigma\in B(H_{S})$, we define the optimized quantum $f$-divergence as \begin{align*}
\tilde{Q}_{f}(\rho\|\sigma)=\sup_{\tau>0,\tr(\tau)\leq 1,\epsilon >0}\tilde{Q}_{f}(\rho\|\sigma+\epsilon \Pi_{\sigma}^{\perp};\tau),
\end{align*}
where $\tilde{Q}_{f}(\rho\|\omega;\tau)$ is defined for positive definite $\omega,\tau\in B_{+}(H_{S})$ as 
\begin{align}
\tilde{Q}_{f}(\rho\|\omega;\tau)=\bra{\phi^{\rho}}_{S\hat{S}}f(\tau_{S}^{-1}\otimes \omega^{T}_{\hat{S}})\ket{\phi^{\rho}}_{S\hat{S}}\pl,\label{df1}
\end{align}
\begin{align*}
\ket{\phi^{\rho}}_{S\hat{S}}=(\rho_{S}^{\frac{1}{2}}\otimes I_{\hat{S}})\ket{\Gamma}_{S\hat{S}}\pl.
\end{align*}
In the above, $\Pi_{\sigma}^{\perp}$ denotes the projection onto the kernel of $\sigma$, $H_{S}$ is an auxiliary Hilbert space isomorphic to $H_{S}$,
$$\ket{\Gamma}_{S\hat{S}}=\sum_{i=1}^{|S|}\ket{i}_{S}\ket{i}_{\hat{S}},$$
for orthonormal bases $\{\ket{i}_{S}\}$ and $\{\ket{i}_{\hat{S}}\}$, and the $T$ superscript indicates transpose with respect to the basis $\{\ket{i}_{\hat{S}}\}$.
\end{definition}
Wilde (\cite{Wil18}) found an equivalent formulation of \eqref{df1} for invertible $\sigma$:
$$\tilde{Q}_{f}(\rho\|\sigma;\tau)=\langle \rho^{\frac{1}{2}}, f(\Delta(\sigma,\tau))(\rho^{\frac{1}{2}})\rangle,$$
where $\Delta(\sigma,\tau)(X):=\sigma X \tau^{-1}$ is the relative modular operator. 
The optimized quantum $f$-divergence  can be simplified as 
\begin{align}
\tilde{Q}_{f}(\rho\|\sigma)=\sup_{\tau>0,\tr(\tau)\leq 1}\tilde{Q}_{f}(\rho\|\sigma,\tau),\label{df2}
\end{align}
where $\sigma$ is invertible.
%The equivalent definition of \eqref{df2} follows from the monotonicity of \eqref{df1} and the continuity of $f$.

Recall that a function $f: J\subset\mathbb{R}\to \mathbb{R}$ is  said to be \textit{operator monotone decreasing} if $A\geq B$ for any $A,B\in B_{sa}(H)$ with spectra in $J$ implies $f(A)\leq f(B)$. Now we are ready to state our main theorem in this section. 
\begin{theorem}\label{main1}
Let $\Phi: B(H_{A})\to B(H_{B})$ be a positive trace preserving linear map satisfying the Schwarz inequality 
\begin{align}\Phi^{*}(X)\Phi^{*}(\tau)^{-1}\Phi^{*}(X^{*})\leq \Phi^{*}(X\tau^{-1}X^{*})\label{eq:schwarz}
\end{align}
for any $X\in B(H_{B})$ and  $\tau\in B_{+}(H_{B})$.
Let $f:(0,\infty)\to\mathbb{R}$ be operator monotone decreasing. For $\rho,\sigma\in B_{+}(H_{A})$ and $\Phi(\rho),\Phi(\sigma)\in B_{+}(H_{B})$, we have  
$$\tilde{Q}_{f}(\rho\|\sigma)\geq \tilde{Q}_{f}(\Phi(\rho)\|\Phi(\sigma)).$$ 
\end{theorem}
\begin{proof}Let us define $V_{\rho}:B(H_{B})\to B(H_{A})$ by
\begin{align*}
V_{\rho}(X):=\Phi^{*}(X\Phi(\rho)^{-\frac{1}{2}})\rho^{\frac{1}{2}}.
\end{align*}
A useful observation is $V_{\rho}(\Phi(\rho)^{\frac{1}{2}})=\rho^{\frac{1}{2}}$.
For any $0<\omega$ with $\tr(\omega)\leq 1$, 
let \begin{align}\label{tauomega1}\tau:=\rho^{\frac{1}{2}}\Phi^{*}(\Phi(\rho)^{-\frac{1}{2}}\omega \Phi(\rho)^{-\frac{1}{2}})\rho^{\frac{1}{2}}.\end{align}
Then $\tau$ is invertible since $\Phi^{*}$ is unital. Now we compute the trace of $\tau$:
\begin{align*}
\tr(\tau)=&\tr(\rho\Phi^{*}(\Phi(\rho)^{-\frac{1}{2}}\omega\Phi(\rho)^{-\frac{1}{2}}))\\
=&\tr(\Phi(\rho)\Phi(\rho)^{-\frac{1}{2}}\omega\Phi(\rho)^{-\frac{1}{2}})=\tr(\omega)\leq 1.
\end{align*}
We claim:
\begin{align}\label{eq:isolemma0}
V_{\rho}^{*}\Delta(\sigma,\tau) V_{\rho}\leq \Delta(\Phi(\sigma),\omega)
\end{align}
and 
\begin{align}\label{eq:isolemma}
f(V_{\rho}^{*}\Delta(\sigma,\tau) V_{\rho})\geq f(\Delta(\Phi(\sigma),\omega)).
\end{align}
Indeed, for any $X\in B(H_{B})$,
\begin{align}
&\langle X, V_{\rho}^{*}\Delta(\sigma, \tau) V_{\rho} (X)\rangle\\
=&\langle V_{\rho}(X),\Delta(\sigma,\tau)V_{\rho}(X) \rangle\\
=&\tr(\rho^{\frac{1}{2}}\Phi^{*}(\Phi(\rho)^{-\frac{1}{2}}X^{*})\sigma\Phi^{*}(X\Phi(\rho)^{-\frac{1}{2}})\rho^{\frac{1}{2}}\rho^{-\frac{1}{2}}\Phi^{*}(\Phi(\rho)^{-\frac{1}{2}}\omega \Phi(\rho)^{-\frac{1}{2}})^{-1}\rho^{-\frac{1}{2}} )\\
=&\tr(\sigma\Phi^{*}(X\Phi(\rho)^{-\frac{1}{2}})\Phi^{*}(\Phi(\rho)^{-\frac{1}{2}}\omega \Phi(\rho)^{-\frac{1}{2}})^{-1}\Phi^{*}(\Phi(\rho)^{-\frac{1}{2}}X^{*}))\\
\leq & \tr(\sigma \Phi^{*}(X\omega^{-1} X^{*}))\label{eq:iso1}\\
=&\langle \Phi(\sigma)X\omega^{-1}X^{*}\rangle\\
=&\langle X, \Delta(\Phi(\sigma),\omega)(X) \rangle\pl,
\end{align}
where the inequality \eqref{eq:iso1} follows from the Schwarz inequality \eqref{eq:schwarz}. The inequality \eqref{eq:isolemma} is an immediate application of \eqref{eq:isolemma0} since $f$ is operator monotone decreasing. 
%Wilde (\cite{Wil18}) proved that for positive definite $\sigma$ the definition \eqref{df1} can be rewritten as  $$\tilde{Q}_{f}(\rho\|\sigma;\tau)=\langle \rho^{\frac{1}{2}}, f(\Delta(\sigma,\tau))(\rho^{\frac{1}{2}})\rangle.$$ 
Now we show the monotonicity of \eqref{df1}:
%for any invertible $\sigma$ and $\Phi(\sigma)$, we have
\begin{align}\tilde{Q}_{f}(\rho\|\sigma;\tau)
=&\langle V_{\rho}(\Phi(\rho)^{\frac{1}{2}}), f(\Delta(\sigma,\tau)) ( V_{\rho}(\Phi(\rho)^{\frac{1}{2}}) )\rangle \label{eq:main1}\\
=&\langle \Phi(\rho)^{\frac{1}{2}}, (V_{\rho}^{*}\circ f(\Delta(\sigma,\tau)) \circ V_{\rho}) (\Phi(\rho)^{\frac{1}{2}}) \rangle \\
\geq & \langle \Phi(\rho)^{\frac{1}{2}}, f(V_{\rho}^{*}\Delta(\sigma,\tau)V_{\rho}) (\Phi(\rho)^{\frac{1}{2}}) \rangle\label{eq:main2}\\
\geq& \langle \Phi(\rho)^{\frac{1}{2}}, f(\Delta(\Phi(\sigma),\omega)) (\Phi(\rho)^{\frac{1}{2}}) \rangle\label{eq:main3}\\
=&\tilde{Q}_{f}(\Phi(\rho)\|\Phi(\sigma);\omega) \pl.
\end{align}
The equality \eqref{eq:main1} follows from the observation that $V_{\rho}(\Phi(\rho)^{\frac{1}{2}})=\rho^{\frac{1}{2}}$. The inequality \eqref{eq:main3} follows  from \eqref{eq:isolemma}. The inequality \eqref{eq:main2} is a direct application of the operator Jensen inequality (\cite{hansen2003jensen}) for operator convex functions. (Here we use the facts that $f:(0,\infty)\to \mathbb{R}$ is operator concave  if $f:(0,\infty)\to \mathbb{R}$ is operator monotone (\cite[Theorem III.2]{ando1978topics}, also see the proof  of \cite[Theorem V.2.5]{Bha13}) and that $f$ is operator convex if and only if $-f$ is operator concave.)
Then we have the monotonicity:
\begin{align}
\tilde{Q}_{f}(\rho\|\sigma)\geq &\sup_{\tau\text{ defined by \eqref{tauomega1}}}\tilde{Q}_{f}(\rho\|\sigma;\tau)\label{eq:main4}\\
\geq & \sup_{\omega>0,\tau(\omega)\leq 1} \tilde{Q}_{f}(\Phi(\rho)\|\Phi(\sigma),\omega)\label{eq:main5}\\
=&\tilde{Q}_{f}(\Phi(\rho)\|\Phi(\sigma))\pl,\label{eq:main6}
\end{align}
where we use the equivalent definition \eqref{df2} in \eqref{eq:main4} and \eqref{eq:main6}.  The inequality \eqref{eq:main5} follows from the monotonicity of \eqref{df1}.
\end{proof}

% In the proof, we also use the facts that $f:(0,\infty)\to \mathbb{R}$ is operator concave  if $f:(0,\infty)\to \mathbb{R}$ is operator monotone (\cite[Theorem III.2]{ando1978topics}) and that $f$ is operator convex if and only if $-f$ is operator concave.
 %The condition imposed on $\Phi$ is stronger than the Schwarz inequality, which degrades to the Schwarz inequality for $\sigma=id$. 
The map $\Phi$ is not necessarily a quantum channel. For example, let $\Phi^{*}(\rho)=\rho^{T}$ be the transpose map, then $\Phi$ satisfy the properties in Theorem \ref{main1}. It is well known that $\Phi$ is not completely positive.  It is worth mentioning that any 2-positive map $\Phi^{*}$ satisfies the Schwarz type inequality, which was mentioned in the email from Mark M. Wilde. We need the following fact (\cite{Car10}):
 \begin{fact*} \rm{Let A and C be positive semi-definite matrices and A be invertible. Then $\left(\begin{matrix} A & B\\
 B^{*}& C\end{matrix}\right)\geq 0$ if and only if $C\geq B^{*}A^{-1}B$.}
 \end{fact*}
\noindent By the fact $\left(\begin{matrix} A & B\\ B^{*} & B^{*}A^{-1}B \end{matrix} \right)\geq 0$. For a $2$-positive $\Phi^{*}$,  we have 
$\left(\begin{matrix} \Phi^{*}(A) & \Phi^{*}(B)\\ \Phi^{*}(B^{*}) & \Phi^{*}(B^{*}A^{-1}B) \end{matrix} \right)\geq 0$.
%where we use $\Phi^{*}(B)^{*}=\Phi^{*}(B^{*})$. 
Using the fact again yields $\Phi^{*}(B^{*}A^{-1}B)\geq \Phi^{*}(B^{*})\Phi^{*}(A)^{-1}\Phi^{*}(B)$.
\begin{remark} \rm{In Theorem \ref{main1}, four quantum states $\rho,\sigma,\Phi(\rho),\Phi(\sigma)$ are all required to be invertible since we are using the equivalent definition $\eqref{df2}$. The invertibility is not necessary in the proof of \cite[Proposition 6]{Wil18}, where the proof of monotonicity relies on the Kraus decomposition of a quantum channel.  It will be interesting to generalize our results to non-invertible quantum states and operator convex functions (for example, see \cite{hiai2011quantum}). It also remains open whether we can prove the monotonicity of the optimized quantum $f$-divergence without the Schwarz inequality \eqref{eq:schwarz}.}
\end{remark}

\section{Monotonicity of Petz $\alpha$-R\'enyi divergence}
 Recall that any self-adjoint operator $x$ can be decomposed as the difference of two positive operators $x=x_{+}-x_{-}$. For any positive operators $\rho,\sigma$ and any $\gamma\in\mathbb{R}_{+}$, the (quantum) $\gamma$-hockey-stick divergence is defined by 
\begin{align}
E_{\gamma}(\rho\|\sigma)=\tr\left((\rho-\gamma\sigma)_{+} \right).
\end{align}
The quantum version of $\gamma$-hockey-stick divergence was introduced by \cite{sharma2012strong} and further explored \cite{hirche2023quantuma,hirche2023quantum}.
The monotonicity of $\gamma$-hockey-stick divergence was proved in 
\cite[Lemma 4]{sharma2012strong}, where they actually proved a stronger result as below.
\begin{lemma}[Monotonicity of hockey-stick divergence] \label{lemma:mhs} Let $\Phi:B(H_{A})\to B(H_{B})$ be a positive trace preserving map. Then for positive operators $\rho,\sigma\in B(H_{A})$ and any $\gamma\in\mathbb{R}_{+}$, we have 
\begin{align*}
E_{\gamma}(\Phi(\rho)\|\Phi(\sigma))\leq E_{\gamma}(\rho\|\sigma)\pl.
\end{align*}
\end{lemma}
\begin{proof} The proof is from \cite[Lemma 4]{sharma2012strong}, and we include it for completeness. Let $P$ be the projection onto the positive part of $\Phi(\rho-\gamma\sigma)$, and we have
\begin{align}
E_{\gamma}(\rho\|\sigma)&=\tr\left((\rho-\gamma\sigma)_{+}\right)\\
&=\tr(\Phi((\rho-\gamma\sigma)_{+}))\label{eq:mhs1}\\
&\geq \tr(P(\Phi((\rho-\gamma\sigma)_{+})))\label{eq:mhs2}\\
&\geq \tr(P(\Phi((\rho-\gamma\sigma)_{+}))-P(\Phi((\rho-\gamma\sigma)_{-})))\label{eq:mhs3}\\
&= \tr(P(\Phi((\rho-\gamma\sigma))))\\
&=\tr((\Phi(\rho)-\gamma\Phi(\sigma))_{+})\label{eq:mhs4}\\
&=E_{\gamma}(\Phi(\rho)\|\Phi(\sigma))\pl.
%E_{\gamma}(\Phi(\rho)\|\Phi(\sigma))&=\tr\left((\Phi(\rho)-\gamma\Phi(\sigma))_{+}\right)=\tr((\Phi(\rho-\gamma\sigma))_{+} )=\tr(P(\Phi(\rho-\gamma\sigma)))\\
\end{align}
The equality \eqref{eq:mhs1} follows from that $\Phi$ is trace preserving. Inequalities \eqref{eq:mhs2} and \eqref{eq:mhs3} follows from the positivity of $\Phi$.The equality \eqref{eq:mhs4} follows from the definition of the projection $P$.  We shall emphasize that complete positivity of $\Phi$ is not necessary in the proof. 
\end{proof}
In \cite{hirche2023quantum}, Hirche and Tomamichel defined the following quantum $f$-divergence.
\begin{definition}[Definition 2.3 in \cite{hirche2023quantum}] We denote by $\mathcal{F}$ the set of functions $f:(0,\infty)\to \mathbb{R}$ that are convex and twice differentiable with $f(1)=0$. Then for any quantum states $\rho,\sigma$, we define the quantum $f$-divergence by 
\begin{align*} 
D_{f}(\rho\|\sigma):=\int_{1}^{\infty} f''(\gamma) E_{\gamma}(\rho\|\sigma)+\gamma^{-3}f''(\gamma^{-1})E_{\gamma}(\sigma\|\rho)d\gamma\pl
\end{align*}
whenever the integral is finite and $D_{f}(\rho\|\sigma)=\infty$ otherwise.
\end{definition}
The quantum $f$-divergence  reduces to Umegaki divergence \begin{align*}D_{f}(\rho\|\sigma):=\begin{cases}\tr(\rho\ln\rho-\rho\ln\sigma),&\text{if 
 } supp(\rho)\subset supp(\sigma)\pl,\\
\infty,&\text{otherwise}\pl,\end{cases}\end{align*} with $f(x)=x\ln(x)$. This integral representation of Umegaki divergence first appeared in \cite{frenkel2023integral}. Another important example is quantum Hellinger divergence 
\begin{align}
H_{\alpha}(\rho\|\sigma):=\alpha\int_{1}^{\infty} \gamma^{\alpha-2}E_{\gamma}(\rho\|\sigma)+\gamma^{-\alpha-1} E_{\gamma}(\sigma\|\rho)d\gamma
\end{align}
with $f(x)=\frac{x^{\alpha}-1}{\alpha-1}$ with $\alpha\in(0,1)\cup (1, \infty)$.
As an immediate application of Lemma \ref{lemma:mhs}, we have the following monotonicity of quantum $f$-divergence.
\begin{theorem}[Monotonicity of quantum $f$-divergence] Let $\Phi:B(H_{A})\to B(H_{B})$ be a positive trace preserving map. For $f\in\mathcal{F}$ and any quantum states $\rho,\sigma\in B(H_{A})$, we have 
\begin{align*}
D_{f}(\Phi(\rho)\|\Phi(\sigma))\leq D_{f}(\rho\|\sigma)\pl.
\end{align*}
\end{theorem}

%\section*{Declaration of data availability and conflict of interest}  Data sharing not applicable to this article as no datasets were generated or analysed during the current study. The author states that there is no conflict of interest.

\section*{Acknowledgement} The author acknowledges support by the DFG cluster of excellence 2111 (Munich Center for Quantum Science and Technology). The author thanks Felix Leditzky, Li Gao, Mark M. Wilde, and Haonan Zhang for reading my note and giving useful comments. The author also thanks Mark M. Wilde for posting the question on Twitter.
%\bibliography{opt}

\begin{thebibliography}{10}

\bibitem{ando1978topics}
T.~Ando.
\newblock Topics on operator inequalities.
\newblock {\em Lecture notes}, Hokkaido Univ., Sapporo, 1978.

\bibitem{beigi2013sandwiched}
Salman Beigi.
\newblock Sandwiched r{\'e}nyi divergence satisfies data processing inequality.
\newblock {\em Journal of Mathematical Physics}, 54(12), 2013.

\bibitem{Bha13}
Rajendra Bhatia.
\newblock {\em Matrix analysis}, volume 169.
\newblock Springer Science \& Business Media, 2013.

\bibitem{Car10}
Eric Carlen.
\newblock Trace inequalities and quantum entropy: an introductory course.
\newblock {\em Entropy and the quantum}, 529:73--140, 2010.

\bibitem{frenkel2023integral}
P{\'e}ter~E Frenkel.
\newblock Integral formula for quantum relative entropy implies data processing
  inequality.
\newblock {\em Quantum}, 7:1102, 2023.

\bibitem{GW20}
Li~Gao and Mark~M Wilde.
\newblock Recoverability for optimized quantum f-divergences.
\newblock {\em Journal of Physics A: Mathematical and Theoretical},
  54(38):385302, 2021.

\bibitem{hansen2003jensen}
Frank Hansen and Gert~K Pedersen.
\newblock Jensen's operator inequality.
\newblock {\em Bulletin of the London Mathematical Society}, 35(4):553--564,
  2003.

\bibitem{hiai2011quantum}
Fumio Hiai, Mil{\'a}n Mosonyi, D{\'e}nes Petz, and C{\'e}dric B{\'e}ny.
\newblock Quantum f-divergences and error correction.
\newblock {\em Reviews in Mathematical Physics}, 23(07):691--747, 2011.

\bibitem{hirche2023quantuma}
Christoph Hirche, Cambyse Rouz{\'e}, and Daniel~Stilck Fran{\c{c}}a.
\newblock Quantum differential privacy: An information theory perspective.
\newblock {\em IEEE Transactions on Information Theory}, 2023.

\bibitem{hirche2023quantum}
Christoph Hirche and Marco Tomamichel.
\newblock Quantum {R\'enyi} and $f$-divergences from integral representations.
\newblock {\em arXiv preprint arXiv:2306.12343}, 2023.

\bibitem{lesniewski1999monotone}
Andrew Lesniewski and Mary~Beth Ruskai.
\newblock Monotone riemannian metrics and relative entropy on noncommutative
  probability spaces.
\newblock {\em Journal of Mathematical Physics}, 40(11):5702--5724, 1999.

\bibitem{muller2017monotonicity}
Alexander M{\"u}ller-Hermes and David Reeb.
\newblock Monotonicity of the quantum relative entropy under positive maps.
\newblock In {\em Annales Henri Poincar{\'e}}, volume~18, pages 1777--1788.
  Springer, 2017.

\bibitem{muller2013quantum}
Martin M{\"u}ller-Lennert, Fr{\'e}d{\'e}ric Dupuis, Oleg Szehr, Serge Fehr, and
  Marco Tomamichel.
\newblock On quantum r{\'e}nyi entropies: A new generalization and some
  properties.
\newblock {\em Journal of Mathematical Physics}, 54(12), 2013.

\bibitem{Petz86}
D{\'e}nes Petz.
\newblock Quasi-entropies for finite quantum systems.
\newblock {\em Reports on mathematical physics}, 23(1):57--65, 1986.

\bibitem{sharma2012strong}
Naresh Sharma and Naqueeb~Ahmad Warsi.
\newblock On the strong converses for the quantum channel capacity theorems.
\newblock {\em arXiv preprint arXiv:1205.1712}, 2012.

\bibitem{tomamichel2009fully}
Marco Tomamichel, Roger Colbeck, and Renato Renner.
\newblock A fully quantum asymptotic equipartition property.
\newblock {\em IEEE Transactions on information theory}, 55(12):5840--5847,
  2009.

\bibitem{um62}
Hisaharu Umegaki.
\newblock Conditional expectation in an operator algebra, iv (entropy and
  information).
\newblock In {\em Kodai Mathematical Seminar Reports}, volume~14, pages 59--85.
  Department of Mathematics, Tokyo Institute of Technology, 1962.

\bibitem{vedral2002role}
Vlatko Vedral.
\newblock The role of relative entropy in quantum information theory.
\newblock {\em Reviews of Modern Physics}, 74(1):197, 2002.

\bibitem{Wil18}
Mark~M Wilde.
\newblock Optimized quantum f-divergences and data processing.
\newblock {\em Journal of Physics A: Mathematical and Theoretical},
  51(37):374002, 2018.

\bibitem{wilde2014strong}
Mark~M Wilde, Andreas Winter, and Dong Yang.
\newblock Strong converse for the classical capacity of entanglement-breaking
  and hadamard channels via a sandwiched r{\'e}nyi relative entropy.
\newblock {\em Communications in Mathematical Physics}, 331:593--622, 2014.

\end{thebibliography}
\bibliographystyle{plain}

\end{document}